\documentclass[10]{article}
\usepackage{arxiv}
\usepackage{layout}
\usepackage[utf8]{inputenc}  
\usepackage[backend=biber]{biblatex} 
\usepackage[colorlinks=true, urlcolor=blue, linkcolor=blue, citecolor=blue]{hyperref}

\usepackage{delimset, diffcoeff, interval}
\usepackage{amsmath, amsthm, amsfonts, mathtools, nicematrix, bm, cancel, siunitx, tensor, bbm}

\usepackage[capitalise]{cleveref}

\usepackage{algorithm}
\usepackage{algpseudocode}

\usepackage{pgf}
\usepackage{pgfplots}

\usepackage{graphicx, booktabs, accents, enumerate}
\usepackage{import, float, pdfpages, transparent, xcolor}
\usepackage[hang]{footmisc} 
\usepackage{authblk}
\usepackage{printlen} 
\usepackage{mwe} 
\usepackage{graphbox}
\usepackage{caption}
\usepackage{subcaption}
\usepackage{tikz}
\usepackage{tikz-cd}

\usetikzlibrary{positioning}

\theoremstyle{definition}

\theoremstyle{plain}
\newtheorem{theorem}{Theorem}[section]

\newtheorem{proposition}[theorem]{Proposition}
\theoremstyle{remark}
\newtheorem{example}[theorem]{Example}

\newtheorem{remark}[theorem]{Remark}

\intervalconfig{soft  open  fences ,separator  symbol=;,}

\newcommand{\defn}{\coloneqq} 

\newcommand{\reals}{\mathbb{R}} 
\newcommand{\com}{\mathbb{C}} 

\newcommand{\sphere}{\mathbb{S}}

\newcommand{\eps}{\epsilon}
\newcommand{\Iddd}{I_{d \times d}}

\newcommand{\gauss}{\mathcal{N}} 

\newcommand{\like}{\mathcal{L}}
\newcommand{\order}{\mathcal{O}}

\newcommand{\re}{\mathfrak{Re}}
\newcommand{\logit}{\mathrm{logit}}
\newcommand{\hyperc}{\mathbb{H}}

\newcommand{\fisher}{\mathcal{I}}

\newcommand{\showfontsize}{\f@size{} pt}

\newcommand{\ibar}{%
  \text{\ooalign{\hidewidth -\kern-.1em-\hidewidth\cr$i$\cr}}%
}

\DeclareMathOperator{\Law}{Law}

\DeclareMathOperator{\diag}{diag}

\DeclareMathOperator*{\expv}{\mathbb{E}} 

\newcolumntype{+}{>{\global\let\currentrowstyle\relax}}
\newcolumntype{^}{>{\currentrowstyle}}

\newtagform{roman}[]()

\usepackage[normalem]{ulem}

\allowdisplaybreaks

\begin{document}


\title{Efficient Bayesian Sampling with Langevin Birth-Death Dynamics}

\author[1]{Alex Leviyev}
\author[2]{Francesco Iacovelli}
\author[1]{Aaron Zimmerman}
\affil[1]{Center for Gravitational Physics, University of Texas at Austin}
\affil[2]{Department of Physics and Astronomy, Johns Hopkins University}
\date{\today}
\maketitle

\begin{abstract}
Bayesian inference plays a central role in scientific and engineering applications by enabling principled reasoning under uncertainty. 
However, sampling from generic probability distributions remains a computationally demanding task. 
This difficulty is compounded when the distributions are ill-conditioned, multi-modal, or supported on topologically non-Euclidean spaces. 
Motivated by challenges in gravitational wave parameter estimation, we propose simulating a Langevin diffusion augmented with a birth-death process. 
The dynamics are rescaled with a simple preconditioner, and generalized to apply to the product spaces of a hypercube and hypertorus. 
Our method is first-order and embarrassingly parallel with respect to model evaluations, making it well-suited for algorithmic differentiation and modern hardware accelerators. 
We validate the algorithm on a suite of toy problems and successfully apply it to recover the parameters of GW150914---the first observed binary black hole merger. 
This approach addresses key limitations of traditional sampling methods, and introduces a template that can be used to design robust samplers in the future. 
Our code is available at \cite{repo}. 

\end{abstract}


\section{Introduction}

Bayesian inference provides a principled framework for estimating unknown quantities from observed data \cite{von2011bayesian,allmaras2013estimating}.
At its core, Bayesian inference involves sampling from a posterior distribution over model parameters constructed from data, prior information, and a forward model.
However, in many scientific applications, this posterior distribution exhibits pathological features that make sampling difficult.
One domain where these difficulties are particularly pronounced is gravitational wave (GW) parameter estimation \cite{bailes2021gravitational}. 
GW signals, emitted by astrophysical events such as binary black hole or neutron star mergers, are detected as faint perturbations by ground-based observatories like LIGO and Virgo \cite{aasi2015advanced,acernese2014advanced, krolak2021recent}. 
From these noisy signals, the task is to infer physical properties of the source---including the component masses, spins, sky location, and distance---by matching the observed data against theoretical waveform models \cite{Thrane_2019,Christensen2022}. 

Accurate inference of these parameters enables tests of general relativity \cite{yunes2013gravitational}, places constraints on the neutron star equation of state \cite{annala2018gravitational}, and provides insights into astrophysical processes (e.g, those that govern binary formation and evolution \cite{abbott2023population}).
However, the problem is computationally challenging due to the slow convergence of sampling algorithms, which require millions of expensive forward models queries.
Current analyses using state-of-the-art techniques often take several hours to days per event, even with large-scale computational resources.
As detector sensitivity increases and the detection rate grows, the computational burden is expected to become significantly more demanding.
This motivates the development of faster, more scalable inference methods that can keep pace with the growing volume of GW data.

GW parameter estimation is a source of many intriguing issues from an algorithm design perspective.
For instance, GW posteriors are supported on the product of three topological spaces.
Orientation parameters such as spins or sky location belong to the sphere $\sphere^2$, coordinates that are prior bounded belong to the hypercube $\hyperc^m$, and periodic coordinates such as signal phases belong to the hypertorus $\mathbb{T}^n$.
Moreover, the posteriors are of moderately high-dimension, and exhibit (pathological) features such as multi-modality \cite{roulet2022removing} and ill-conditioning.
Coupled together, these issues make estimating parameters from GW data a formidable and unique challenge.

A variety of sampling algorithms have been adopted for GW parameter estimation.
Nested sampling remains one of the most popular approaches \cite{skilling2006nested,veitch2015parameter,Ashton-2019}, with implementations such as \texttt{CPNest} \cite{cpnest}, \texttt{PyMultiNest} \cite{feroz2009multinest}, and the widely used \texttt{Dynesty} library \cite{Speagle-2020}.
Although robust, nested sampling is computationally expensive, especially for high-dimensional problems with pathological features.
Markov chain Monte-Carlo (MCMC) methods are also extensively used, coupled with advanced techniques like parallel tempering \cite{earl2005parallel}, affine-invariant ensemble sampling \cite{goodman2010ensemble}, and custom jump proposals \cite{veitch2015parameter, biwer2019pycbc}.
However, many MCMC algorithms are zeroth order, meaning they only query the forward model, and not its gradient.
Due to the recent ports of GW models to differentiable programming languages \cite{Iacovelli-2022,edwards2023rippledifferentiablehardwareacceleratedwaveforms}, sampling methods utilizing gradient information also are under active development \cite{wong2023fastgravitationalwaveparameter}.
However, these have only recently been applied to GW parameter estimation, and have not yet been widely adopted.
To date, a diffusion based sampling method robust to the issues presented above which can be efficiently implemented on modern hardware accelerators has not yet been developed.
Our goal in this paper is to propose such a method using Langevin diffusion.

In \cref{sec:ula}, we review Langevin dynamics, and the idea of \emph{ensemble} particle diffusions.
Ensemble diffusions offer many advantages in constrast to single particle methods, such as the ability to construct preconditioners and more robustly discover modes.
In \cref{sec:fisher}, we propose a simple way to calculate the optimal Fisher preconditioner introduced in \cite{titsias2024optimal} for ensemble methods.
We show that incorporating this preconditioner significantly improves convergence rates on ill-conditioned test problems. 
In \cref{sec:reparam} we thoroughly investigate the procedure of inheriting flows over $\hyperc^m$ from flows over $\reals^m$ via reparameterization.
Using the framework we develop, we provide guidelines on how to construct reparameterized diffusions with desirable behavior. 
We also illustrate that reparameterization couples non-trivially with well known techniques such as preconditioning and annealing.
Finally, in \cref{sec:birth-death}, we discuss our implementation of the birth-death process introduced in \cite{lu2019accelerating, PhysRevE.107.024141}, which is crucial for accurately reconstructing multi-modal distributions. 
All the tools introduced are generalized to be applicable on the hypercube, the hypertorus, and products thereof, and validated with experiments in \cref{sec:experiments}.
In \cref{sec:gw150914} we conclude with a proof of principle by estimating the parameters of GW150914 \cite{Abbott2016,PhysRevLett.116.241102}, the first detected binary black-hole coalesence.

\section{Background}

Let us consider the concrete application of parameter estimation. 
In this setting, one has a posterior $p : \chi \to \mathbb{R}^+$ supported on a (possibly non-Euclidean) space $\chi$, and the task is to draw a large quantity of independent and identically distributed (i.i.d.) samples from it. 
A large i.i.d. sample set from a probability distribution is practically equivalent to a complete characterization of $p$. 
Indeed, the samples provide information regarding the morphology of $p$, including the presence of correlations or multimodality. 
On the other hand, the samples can be used to construct well-behaved estimators of expectation values with respect to $p$ via the central limit theorem \cite{weinan2021applied}. 
Unfortunately, drawing i.i.d. samples is impossible in most modern applications, and can only be done for a handful of classical probability distributions. 
Hence, we must resort to approximate methods to draw representative samples.

\subsection{Ensemble Langevin Diffusion}
\label{sec:ula}
The \textit{Langevin diffusion} \cite{pavliotis2014stochastic} is a stochastic process $X_t$ defined over $\reals^d$ which obeys the following stochastic differential equation:
\begin{align}
dX_t = - \nabla V(x) dt + \sqrt{2} dB_t,
\end{align}
where $V:\reals^d \to \reals$ is the \textit{potential}, and $B_t$ is a standard Brownian motion.
The density corresponding to $X_t$, denoted by $\rho_t := \Law X_t$, in this case obeys the \textit{Fokker-Planck equation}:
\begin{align} 
  \label{eq:FPE}
\partial_t \rho_t = \nabla \cdot (\nabla \rho_t + \rho_t \nabla V),
\end{align}
and that under appropriate conditions on the potential that density $\rho_t \to e^{-V}$ as $t \to \infty$ \cite{pavliotis2014stochastic}, i.e, the density approaches the Gibbs distribution, also referred to as the \textit{target}.
This mechanism may be exploited to sample from probability distributions.
Suppose we set $V(x) = - \ln p(x)$, where $p$ is a Bayesian posterior.
Then the Langevin diffusion yields an evolution of $X_t$ such that $\rho_t \to p$ as $t \to \infty$.
Hence, asymptotically the evolution of $X_t$ should yield correlated samples from $p$.
This process may be simulated by a variety of discretization schemes, the simplest of which is known as the \textit{Euler-Maruyama} discretization:
\begin{align}
\label{eq:standard-langevin}
x_{l+1} \leftarrow x_l - \nabla V(x_l) \tau + \sqrt{2 \tau} \xi, \quad \xi \sim \gauss(0, \Iddd).
\end{align}
This discretization is also referred to as the unadjusted Langevin algorithm (ULA) \cite{chewi2025_logconcave_sampling}.\footnote{This is in contrast to the Metropolis adjusted Langevin algorithm (MALA), which incorporates a Metropolis filter \cite{robert2004metropolis}.}
One issue with \cref{eq:standard-langevin} is that the convergence of the dynamics is highly sensitive to the conditioning of the target.
To address this issue, curvature information can be incorporated into the dynamics to adjust for anisotropies \cite{ma2015completerecipestochasticgradient}:
\begin{align}
\label{eq:fisher-langevin}
x_{l+1} \leftarrow x_l - \mathcal{I}^{-1}\nabla V(x_l) \tau + \sqrt{2 \tau } \eta, \quad \eta \sim \gauss(0, \mathcal{I}^{-1}),
\end{align}
where $\mathcal{I} \in \reals^{d \times d}$ is a constant positive definite matrix specified in \cref{sec:fisher}. 
In \cref{app:hessian-reparam} we discuss the extension to quasi-Newton methods, which allow for preconditioners with spatial variation.

We may efficiently draw from the noise term in \cref{eq:fisher-langevin} with the help of the following
\begin{proposition}
Let $\xi \sim \mathcal{N}(0, I)$, $\mathcal{I} \in \reals^{d \times d}$ be a positive definite matrix, and $U$ is the upper triangular Cholesky decomposition of $\mathcal{I}$, so that $\mathcal{I} = U^\top U$.
Then $U^{-1} \eta \sim \mathcal{N}(0, \mathcal{I}^{-1}).$
\end{proposition}

\begin{proof}
$\mathcal{N}(0, \mathcal{I}^{-1}) = \mathcal{N}(0, (U^\top U)^{-1})
= \mathcal{N}(0, U^{-1} (U^\top)^{-1})
= \mathcal{N}(0, U^{-1} (U^{-1})^{\top})
= U^{-1} \mathcal{N}(0, I).$
\end{proof}
Hence, the noise term can be calculated with a triangular solve after the decomposition is computed.
In the context of GW astronomy, $d \sim \mathcal{O}(10)$, and hence the decomposition is not prohibitive.
The computed decomposition may be recycled to solve for the drift term in \cref{eq:fisher-langevin} as well.

The ULA simulates a \textit{single} particle, however there are good reasons to consider the extension to the multi-particle (ensemble) case.  
Simulating an ensemble accrues minimal expense if one has access to parallel computing hardware.
The most costly aspect of simulating the dynamics, the evaluation of the gradient of the potential, can be accomplished in an embarassingly parallel fashion.
This opens the door to many other advantages that ensemble based methods enjoy, such as more robust exploration of the parameter space \cite{Syed2021,lu2019accelerating}, and the ability to construct preconditioners \cite{leimkuhler2018ensemble}.

\subsection{Optimal Fisher Preconditioning}
\label{sec:fisher}
Both gradient descent and Langevin dynamics are not affine invariant and are thus sensitive to linear scalings such as choice of units (e.g, meters vs kilometers).
This is an undesirable property in engineering and physics applications, as it is not clear a priori which scalings are best for numerics.
Similarly, strong correlations between parameters lead to directions in the energy landscape with vastly different length scales, posing a challenge to gradient based methods \cite[Chapter~2]{liquet2024mathematical}.
For these reasons an appropriate rescaling of the geometry can greatly accelerate convergence.

There are many strategies in the literature suggesting how to do this in practice, each prescribing a matrix field $A: \reals^d \to \reals^{d \times d}$ which adjusts the gradient and noise terms in the Langevin dynamics \cite{ma2015completerecipestochasticgradient}. 
For example, a simple choice is $A = \diag(\tau_1, \ldots, \tau_d)$, which implements a different time-step in each direction \cite{porter2014hamiltonian}.
On the other hand, $A(x) \approx (\nabla^2 V(x))^{-1}$ corresponds to quasi-Newton type diffusions \cite{martin2012stochastic}.
Recently \cite{titsias2024optimal} discovered that there exists a constant preconditioning matrix that maximizes the expected squared jump distance of the Langevin dynamics.
This result is interesting as it generalizes the intuition that a good preconditioning matrix estimates the covariance of the target, and can be trivially calculated in the context of ensemble based methods.

Let us define a cost function $J$, denoting the expected squared jump distance of the Langevin dynamics with timestep $\tau>0$ and preconditioning matrix $A \in \reals^{d \times d}$
\begin{equation}
  \label{eq:expected-jump-cost}
J(\tau, A) = \expv \norm{x_{l + 1} - x_l} ^ 2 = \norm{- \tau A \nabla V(x_l) + \sqrt{2 \tau A} \gauss(0, \Iddd)}^2.
\end{equation}
Then for a fixed timestep \cite{titsias2024optimal} showed that \cref{eq:expected-jump-cost} attains its maximal value for $A^* \propto \fisher^{-1}$, where
\begin{equation}
\fisher := \expv_{x \sim p} \nabla V(x)  \nabla V(x)^{\top}, 
\label{eq:fisher}
\end{equation}
is referred to as the \textit{Fisher covariance matrix}.
\begin{remark}
The name ``Fisher matrix'' is an overloaded term.
In the GW literature, the Fisher matrix is commonly taken to be the Gauss-Newton approximation of the Hessian of the Bayesian potential evaluated at the maximum a-posteriori point (see \cref{app:quasi} for details).
We refer to \cref{eq:fisher} as the Fisher matrix exclusively.
\end{remark}

\cite{titsias2024optimal} proposes calculating $\fisher$ using the history of a single particle.
In the ensemble case, we propose a simple alternative utilizing a Monte-Carlo approximation.
Observe that the Fisher matrix requires taking an expectation value with respect to the target $p$, which we are unable to calculate analytically.
However, as the diffusion evolves, we expect the samples of the ensemble to begin resembling draws from $p$.
Hence, we can approximate \cref{eq:fisher} with the following:
\begin{equation}
  \label{eq:fisher-approx}
  \mathcal{I} \approx \frac{1}{N} \sum_{n=1}^N \nabla V(x_n) \nabla V(x_n)^{\top} + \lambda \Iddd,
\end{equation}
where we have added a small damping term with $\lambda>0$ for numerical stability.
In the ensemble Langevin case, the gradients are reused to compute $\mathcal{I}$, and then used to adjust the drift and noise terms as discussed in \cref{sec:ula}.
The preconditioner defined in \cref{eq:fisher-approx} is likely to be useful in other popular ensemble methods approaches, such as Stein variational gradient descent \cite{liu2016stein,wang2019steinvariationalgradientdescent}.
Note that $\mathcal{I}$ will certaintly be outperformed by pointwise preconditioners.
However, \cref{eq:fisher-approx} is a convenient choice in this setting, and more sophisticated techniques may be employed as the need arises.

We emphasized in the beginning of \cref{sec:ula} that the Langevin dynamics is \textit{strictly} defined only over $\reals^d$.
However, our target application requires us to consider more general spaces.
In the following section we discuss how a flow over a hypercube may be inherited from a flow defined over $\reals^d$ given a suitable coordinate transformation.

\subsection{Reparameterization}
\label{sec:reparam}

A hypercube support space $\mathbb{H}^d := \bigtimes_{i=1}^d (a_i, b_i)$, where $a_i < b_i \in \reals$ appear naturally in many scientific and engineering applications \cite{hsieh2020mirroredlangevindynamics}.
This is a consequence of the fact that parameters are oftentimes constrained by definition or prior bounded.
However, many processes of interest (e.g, Langevin dynamics) are defined over $\mathbb{R}^d$, and are thus ill-defined in their standard forms over constrained spaces.
In this section we describe a way to inherit a flow over $\hyperc^d$ given a flow over $\reals^d$, and a suitable coordinate transformation $T: \mathbb{H}^d \rightarrow \mathbb{R}^d$.
We provide a recipe to design such transformations, and state a result which allows one to calculate the pushforward potential and its gradient in a simple way.

For simplicity we consider transformations $T: \hyperc^d \to \reals^d$ that act on each coordinate of a point $x \in \hyperc^d$ independently.
In other words, the map can be specified by $(T(x))_i = \hat{T}_i(x_i)$, where each component map $\hat{T}_i: [a_i, b_i] \rightarrow \mathbb{R}$ is responsible for pushing forward coordinate $1 \le i \le d$ independently of the other coordinates.
To specify these $\hat{T}_i$ coordinate maps, we will need an affine transformation $\sigma^{-1}_i: [a_i, b_i] \rightarrow [0,1]$ for each coordinate $i$ which is responsible for ``standardizing'' each bounded interval. 
This is accomplished by $\sigma_i(x) := (b_i - a_i) x + a_i$.
With this affine transformation in hand, we can construct the component transformations by selecting a map $q: [0,1] \rightarrow \mathbb{R}$, and forming the composition $q \circ \sigma_i^{-1} : [a_i, b_i] \to \reals$.
A convenient choice for $q$ is the \textit{quantile function} associated to a random variable taking values in $\mathbb{R}$, which is defined as the inverse of the cumulative distribution function (CDF) $F: \mathbb{R} \rightarrow [0,1]$ of that random variable.
The proposed recipe is summarized in the commutative diagram illustrated in \cref{fig:coordinate_transform}.

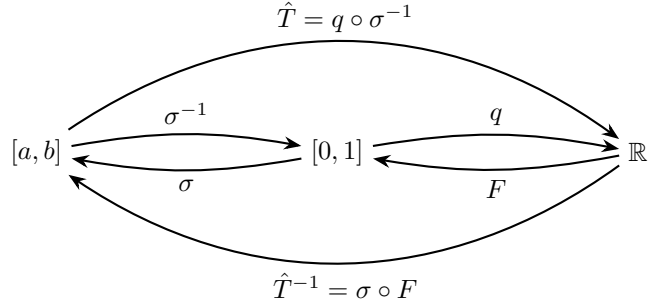
\begin{figure}[t]
    \centering
\begin{tikzpicture}[node distance=2cm, auto]

    \node (a) at (0,0) {$[a,b]$};
    \node (b) at (4,0) {$[0,1]$};
    \node (c) at (8,0) {$\mathbb{R}$};
  
    \draw[-{Stealth}, thick] (a) to[bend left=10] node[above] {$\sigma^{-1}$} (b);
    \draw[-{Stealth}, thick] (b) to[bend left=10] node[above] {$q$} (c);
    \draw[-{Stealth}, thick] (c) to[bend left=10] node[below] {$F$} (b);
    \draw[-{Stealth}, thick] (b) to[bend left=10] node[below] {$\sigma$} (a);
  
    \draw[-{Stealth}, thick, bend left=35] (a) to node[above] {$\hat{T} = q \circ \sigma^{-1}$} (c);
    \draw[-{Stealth}, thick, bend left=35] (c) to node[below] {$\hat{T}^{-1} = \sigma \circ F$} (a);
  
\end{tikzpicture}
\caption{Diagram showing the composition of maps defining $\hat{T}$ and its inverse over the interval $[a,b]$.}
\label{fig:coordinate_transform}
\end{figure}

We now state the main result.

\begin{theorem}    
  \label{thm:reparam}
Let $V$ be a potential with $\operatorname{supp} V = \mathbb{H}^d = \bigtimes_{i=1}^d [a_i, b_i]$. 
Let $\eta : \Omega \to \mathbb{R}$ be a continuous random variable with PDF $f$, and $T : \mathbb{H}^d \to \mathbb{R}^d$ be a coordinate transformation that is constructed via the above described procedure. 
Then the pushforward of the potential $V_{\#}: \reals^d \to \reals$ is given by
\begin{equation}
    V_{\#}(y) = V(T^{-1}(y)) + U(y) - \ln \operatorname{vol} \mathbb{H}^d,
\end{equation}
where $U(y):= \sum_{i=1}^d -\ln f(y_i)$.
Consequently the gradient of the pushforward potential is given by
\begin{equation}
    \partial_{y_j} V_{\#}(y) = \partial V(T^{-1}(y)) \, \Delta_j f(y_j) + \partial_{y_j} U(y),
\end{equation}
where $\Delta_j := b_j - a_j$.
\end{theorem}

\begin{proof}
    Recall that the pushforward density is given by
    \[
    p_\#(y) = p(T^{-1}(y)) \left| \det \nabla_y T^{-1}(y) \right|,
    \]
    and therefore the pushforward potential takes the form
    \[
    V_\#(y) = V(T^{-1}(y)) - \ln \left| \det \nabla_y T^{-1}(y) \right|.
    \]
    Simplifying the second ``correction" term, we have
    \begin{align*}
    \ln \left| \det \nabla_y T^{-1}(y) \right|
    &= \ln \left| \prod_{i=1}^d \partial_{y_i} \sigma_i(F(y_i)) \right| \\
    &= \ln \left[ \prod_{i=1}^d \sigma_i'(F(y_i)) \partial_{y_i} F(y_i) \right] \\
    &= \ln \left[ \prod_{i=1}^d \Delta_i f(y_i) \right] \\
    &= \ln \operatorname{vol}(\hyperc^d) + \sum_{i=1}^d \ln f(y_i).
    \end{align*}
    This proves the expression for the pushforward potential.
    The expression for the gradient of the pushforward potential follows directly.
\end{proof}

\cref{thm:reparam} has an appealing interpretation. 
Namely, the pushforward potential $V_{\#}$ inherits its value from $V$, but is corrected by a ``confining potential'' $U:\reals^d \to \reals$ whose particular form depends on the probability density function (PDF) of the selected random variable.\footnote{There is also a constant term that depends on the volume of the hypercube. However, we are not interested in this constant term as it has no effect on the Langevin dynamics.}
This confining potential ensures that particles are incentivized to remain away from the boundaries of the hypercube. 
On the other hand, the gradient $\nabla V_\#$ inherits its value from $\nabla V$, but is scaled by the PDF, and incorporates a ``restoring force'', ensuring particles are drawn back into bounds.
These considerations allow one to engineer reparameterizations with desirable boundary behavior.


In \cref{app:reparameterization} we provide a more detailed discussion of the implications of \cref{thm:reparam}, and concretely analyze transformations constructed using the Cauchy, Gaussian, and logistic random variables.
We begin with the case of the logistic random variable, which corresponds to the method currently utilized in libraries such as Stan \cite{carpenter2017stan}, and show that it is asymptotically equivalent to Lasso ($l_1$) regularization of the potential \cite{wang2013tikhonov}.
Since the logistic random variable is sub-Gaussian (tails decay faster than Gaussian tails), the next interesting case is to test reparameterizations constructed with the Gaussian.
In this case the particle pushforward and pullback maps are special functions. 
However, as we will see, the Gaussian is unique in that it provides a ``strong'' harmonic confining potential, consequently yielding desirable behavior of flows near the boundary.
Equivalently, a Gaussian can be thought of as Tikhonov ($l_2$) regularization of the potential \cite{wang2013tikhonov}.
Finally, we conclude with a random variable with heavier tails than a Gaussian, the Cauchy.
Results of these experiments will be presented in \cref{sec:experiments}.

\subsection{Birth Death Dynamics}
\label{sec:birth-death}

The Langevin birth-death dynamics describes the evolution of a probability density $\rho_t$ satisfying the following PDE \cite{lu2019accelerating}
\begin{equation}
\label{eq:bd-dynamics}
\partial_t \rho_t = \nabla \cdot (\rho_t \nabla V + \nabla \rho_t) - \alpha(x, \rho_t) \rho_t,
\end{equation}
where the birth-death rate $\alpha$ is defined as
\begin{equation}
\label{eq:bd-alpha}
\alpha(x; \rho, p) := \log \rho(x) - \log p(x) - \mathbb{E}_{x' \sim \rho} \left[ \ln \rho(x') - \ln p(x') \right],
\end{equation}
and $p = e^{-V}$ is the target posterior.
The first term on the right hand side of \cref{eq:bd-dynamics} is simply the evolution prescribed by \cref{eq:FPE}.
The novel ``birth-death'' aspect of the dynamics is given by the second term, which provides a rule for the decay or growth of the density $\rho_t$ at every point $x\in\reals^d$.
Namely, if $\alpha(x) > 0$ corresponding to an ``excess of density'' at $x$, then $\rho_t$ decays at that point.
Likewise, if $\alpha(x) < 0$ corresponding to a ``deficit of density'' at $x$, then $\rho_t$ grows at that point.
If $\alpha(x) = 0$, nothing happens.

\begin{remark}
\cref{eq:bd-dynamics} has several key properties.
First, the rate term $\alpha$ is identically zero when $\rho_t=p$, hence $p$ is a stationary point of the dynamics.
Second, the dynamics are mass-preserving, meaning that if $\rho_0$ is a probability measure, it remains a probability measure along the course of its evolution.
Third, the rate is independent of the normalizing constant of $p$.
Finally, the ``pure'' birth death dynamics are diffeomorphic invariant.
This implies that if $\rho_\#$ evolves according to a birth-death process, then so too does $\rho$, and their rates satisfy the relationship $\alpha(\cdot, \rho_\#, p_\#) = \alpha(T^{-1}(\cdot), \rho, p)$.
Notably, the smoothing procedure introduced below in \cref{eq:pampel-lambda} breaks diffeomorphic invariance.
\end{remark}

It was shown in \cite{lu2019accelerating} that these dynamics exhibit a gradient flow structure, and are therefore amenable to associated PDE analysis techniques.
One consequence of the analysis is that the rate of convergence to the unique stationary point $e^{-V}$ (under general conditions) is independent of the particular details of the potential.
This is remarkable, as it suggests that convergence of the dynamics is robust to issues related to ill-conditioning and multimodality; this is in stark contrast to pure Langevin dynamics, where convergence rates are known to strongly depend on such factors.
Thus Langevin birth-death is an appealing method to investigate, since most inference tasks in engineering and physics suffer from ill-conditioning, multimodality, or both \cite{Feroz_2008}.

There are several hurdles that must be overcome in order to simulate \cref{eq:bd-dynamics} with particles.  
First, the rate $\alpha$ must be replaced with an alternative that is well defined for discrete measures.
We select 
\begin{equation}
\label{eq:pampel-lambda}
\Lambda(x; \rho) := \ln k * \frac{\rho}{p}(x) - \mathbb{E}_{x' \sim \rho} \ln k * \frac{\rho}{p}(x'),
\end{equation}
where $*$ denotes the convolution operation and $k$ is a kernel.
This is convenient as it is numerically tractable and identically zero when $\rho = p$.
Second, the jump process proposed to simulate the birth-death dynamics couples the rate calculation with the diffusion \cite{lu2019accelerating}, meaning that each particle in the ensemble must be updated in a serial manner.  
As long as the number of particle jumps per iteration remain low however, the diffusion and jump processes may be decoupled \cite{PhysRevE.107.024141}, leading to an embarassingly parallel method. 
In practice, we ensure that jumps remain rare by adaptively scaling the rate $\Lambda$ by a constant $c > 0$ such that the number of particles accepted to jump is less than some prescribed fraction of the ensemble size (see \cref{app:bounding-jumps}).  
Finally, the performance of the simulation critically depends on selecting an appropriate smoothing kernel $k$.  
We select a Gaussian kernel with a covariance matrix given by the inverse Fisher (\cref{eq:fisher-approx}), and select the bandwidth using the median heuristic \cite{garreau2017large}
\begin{equation}
  \label{eq:rbf-kernel}
k(x, y) = \exp \left( -\frac{1}{2h} \norm{x - y}_{\mathcal{I}}^2 \right),
\end{equation}
where $h=\operatorname{med}_{i \neq j} \norm{x_i - x_j}_{\mathcal{I}}^2$.
The algorithm is summarized in \cref{algo:pampel-bd}.
To implement the algorithm efficiently, we utilize a data structure called \texttt{ParticleTracker} which allows for $\mathcal{O}(1)$ selection, removal, and choice.
Further discussion of the \texttt{ParticleTracker} data structure may be found in \cref{app:pt}.

\begin{algorithm}[t]
    \caption{Pampel Birth-Death} 
    \label{algo:pampel-bd}
    \begin{algorithmic}[1] 
        
        \Require Particle set $X$, constant $\gamma>0$, and rates $\Lambda_i$ for every particle $i$
        
        \State Sample $r_i \sim \mathcal{U}(0,1)$ for all $i$
        \State $\xi \gets$ Random permutation of $\{i : r_i < 1 - \exp(-|\Lambda_i| \cdot \gamma) \}$
        \State $\text{alive} \gets \{1, 2, \dots, N\}$
        \State $\text{jumps} \gets [1, 2, \dots, N]$
        
        \For{$i \in \xi$}
            \If{$i \in \text{alive}$}
                \State $j \gets$ alive.choose\_random\_item()
                \If{$\Lambda_i > 0$}
                    \State $\text{jumps}[i] \gets j$
                    \State alive.remove\_item($i$)
                \Else
                    \State $\text{jumps}[j] \gets i$
                    \State alive.remove\_item($j$)
                \EndIf
            \EndIf
        \EndFor

        \State \Return $X[\text{jumps}]$
        
    \end{algorithmic}
\end{algorithm}

\subsection{Annealing}

The Langevin dynamics are known to exhibit \textit{metastability} \cite{weinan2021applied}.
This refers to the fact that the expected time it takes for a particle to hop over an energy barrier is exponential in the height of the barrier.  
In practical terms this means that once a particle finds a deep mode, it is unlikely to ever leave.
Having many particles explore the energy landscape therefore improves the chances of discovering the primary modes of a posterior.  
However, the landscapes we encounter in practice are formidable, and it becomes prohibitively expensive to simulate the number of particles necessary to get adequate coverage.  
This motivates simulated annealing \cite{neal1996sampling}, a tool used to facilitate a particles exploration of the energy landscape.  

Let $V: \reals^d \to \reals$ be a potential.
We begin by introducing an (inverse) temperature schedule function $\beta: [0,1] \rightarrow [\beta_{\min}, 1]$, where $0 < \beta_{\min} < 1$, and $\lim_{t \to 1} \beta(t) = 1$.
Here, $\beta_{\min}$ denotes the initial temperature the posterior will be heated to, and the schedule returns to ``room temperature'' ($\beta=1$) towards the end of the simulation.  
When $\beta$ is smooth, it is called an \textit{annealing schedule}.  
Then the heated posterior and corresponding potential then take the form  
\begin{align*}
p(x;\beta) &= e^{-\beta V(x)}, & V(x; \beta) &:= -\ln p(x;\beta) = \beta V(x).
\end{align*}

The idea is to perform the dynamics over a heated posterior as it cools down, hence facilitating exploration at the beginning of the simulation.  
The dynamics which implement this idea are given as  
\begin{align*}
    x_{\ell + 1} &= x_{\ell} - \nabla V(x_{\ell};\beta) \, \epsilon + \sqrt{2 \epsilon} \, \xi, \quad \xi \sim \mathcal{N}(0, \Iddd) \\
    &= x_{\ell} - \nabla V(x_{\ell}) \, \beta \, \epsilon + \sqrt{2 \epsilon} \, \xi \\
    &= x_{\ell} - \nabla V(x_{\ell}) \, \tau + \sqrt{\frac{2 \tau}{\beta}} \, \xi,
\end{align*}
where we have introduced a time reparameterization $\tau = \beta \epsilon$.  
Hence, Langevin dynamics on an annealed potential is equivalent to room temperature dynamics with a temperature scaled noise term.  
We seek an analogous result for preconditioned Langevin over a hypercube.  

Let $V : \mathbb{H}^d \to \mathbb{R}$ be a potential with hypercube support, and $T : \mathbb{H}^d \to \mathbb{R}^d$ be a coordinate transformation constructed by the procedure outlined in \cref{sec:reparam}.  
Then by pushing forward the heated potential $V(\cdot \, \,; \beta)$, we get  
\begin{align*}
    T_\#(V(\cdot \, \,; \beta))(y) &= V(T^{-1}(y)) \, \beta + U(y) \\
    &= \beta \left[ V(T^{-1}(y)) + \frac{1}{\beta} U(y) \right] \\
    &=: \beta \, V_\#(y;\beta),
\end{align*}
where we have introduced the temperature dependent potential $V_\#(y; \beta)$.
Notably, this potential has a ``cooled'' confinement $U$, and simplifies to $V_\#(y)$ when $\beta = 1$.  
Then  
\begin{align*}
    y_{\ell + 1} &= y_\ell - \nabla T_\#(V(\cdot \, \,; \beta))(y_\ell) \, \epsilon + \sqrt{2 \epsilon} \, \xi, \quad \xi \sim \mathcal{N}(0, \Iddd) \\
    &= y_\ell - \beta \nabla V_\#(y_\ell; \beta) \, \epsilon + \sqrt{2 \epsilon} \, \xi \\
    &= y_\ell - \nabla V_\#(y_\ell; \beta) \, \tau + \sqrt{\frac{2 \tau}{\beta}} \, \xi,
\end{align*}
where $\tau = \beta \epsilon$ was used.  
We see that, as before, the noise is scaled by the temperature, but now the confining potential is correspondingly ``strengthened''.  

\begin{remark}
Let us contrast this to a naive annealing approach, which would scale the noise term in the dual space.  
Then we have  
\begin{align*}
    y_{\ell + 1} &= y_{\ell} - \nabla V_\#(y_{\ell}) \, \tau + \sqrt{\frac{2 \tau}{\beta}} \, \xi \\
    &= y_\ell - \left[ \beta \nabla V_\#(y_\ell) \right] \epsilon + \sqrt{2 \epsilon} \, \xi,
\end{align*}
from which we can conclude that the naive approach corresponds to Langevin dynamics with potential  
\begin{align*}
    \beta V_\#(y) = \beta V(T^{-1}(y)) + \beta U(y).
\end{align*}

This has the unfortunate property of heating up the confining potential, and making it more likely for particles to visit the extremal regions near the boundary of the hypercube.  
Our intuition tells us that this will likely lead to numerical issues, as particles will be incentivized to visit regions near the boundary of the hypercube.

\end{remark}

\subsection{Toroidal Topology}
\label{sec:toroid}

Our strategy is to treat the support of a GW model as $\chi = \mathbb{H}^8 \times \mathbb{T}^3$.  
We have already discussed how to handle the hypercube.  
In this section we review how to adapt the Langevin diffusion and birth-death to accommodate the hypertorus.  
Let $x \sim y$ be an equivalence relation on $\reals$ which holds true if there exists a $k \in \mathbb{Z}$ such that $x + 2\pi k = y$. 
A classic result from topology states that $\mathbb{R}/\mathord{\sim}$ equipped with the quotient topology is homeomorphic to $\mathbb{S}$.  
This provides a convenient strategy to pushforward objects defined on $\mathbb{R}$ to the circle $\mathbb{S}$ via the following diagram
\begin{center}
\begin{tikzcd}
\mathbb{R} \arrow[d, "\text{mod}"'] \arrow[rd, "\phi \, \circ \, \text{mod} "] & \\
\mathbb{R}/\mathord{\sim} \arrow[r, "\phi"'] & \mathbb{S}
\end{tikzcd}
\end{center}
where $\phi: \theta \mapsto (\cos \theta, \sin \theta)$.  
In practice we stop short of applying $\phi$, and mainly work on $\mathbb{R}/\mathord{\sim}$.
The generalization to the hypertorus follows much the same logic.  

Pushing forward the Langevin dynamics via the diagram yields the \emph{circular Langevin diffusion} \cite{GarcaPortugus2017}, while annealing and preconditioning port over trivially.  
However, the birth-death process depends on the support space through the kernel, and therefore requires modification.  
Our primary concern is to incorporate anisotropies in the kernel between bounded and periodic coordinates.  
We accomplish this in an identical manner, i.e, by pushing forward a Gaussian random variable via the diagram.  
This yields the \emph{wrapped Gaussian} \cite{mardia2009directional} whose density is given by
\begin{equation}
g(x ; \mu, \sigma^2) = \sum_{k \in \mathbb{Z}} \mathcal{N}(x + 2\pi k ; \mu, \sigma^2).
\end{equation}
To obtain a useful kernel, we begin by constructing a kernel over $\mathbb{R}^{11}$ using \cref{eq:rbf-kernel}, and then add a (truncated) sum of shifts in the appropriate coordinates to obtain a kernel over $\mathbb{R}^8 \times \mathbb{R}^3/\mathord{\sim}$.  
The resulting kernel incorporates anisotropy between all coordinates and satisfies periodic boundary conditions.

\subsection{Spherical topology}
\label{sec:spherical-topology}

Spherical topology enters the GW model through the various orientation parameters required to specify the state of the coalescing binary.  
The simplest way to address this scenario, and the one adopted in this work, is to employ a “cylindrical approximation” to the sphere.  
That is, we use a single chart and incorporate periodic boundary conditions on the azimuthal coordinate, $(\theta, \phi) \in [0, \pi) \times \mathbb{S}$.  
This simplification allows us to handle $\mathbb{S}^2$ with the tools developed in \cref{sec:reparam} and \cref{sec:toroid}.  
Note that this simplicity comes at the cost of significant distortions near the poles.  
We leave other approaches to handle spherical topologies to future work.

\section{Numerical experiments}
\label{sec:experiments}

\subsection{Hybrid Rosenbrock}
In this section we introduce the hybrid Rosenbrock density (HRD) \cite{pagani2022n} which we will use to benchmark our methods.
The HRD is a useful model to test sampling schemes for two primary reasons.
First, it is a simple non-convex model that exhibits ``banana'' like features in each marginal. 
This makes it ideal for testing a sampler's robustness with respect to ill-conditioned energy landscapes. 
Second, the hybrid Rosenbrock can be sampled directly.
This yields a convenient way to compare the output of the sampler under question with ``ground truth''.

The \textit{hybrid Rosenbrock} target is given by:

\begin{equation}
  \pi(\mathbf{x}) \propto \exp \left\{ 
      -a (x_1 - \mu)^2 
      - \sum_{j=1}^{n_2} \sum_{i=2}^{n_1} b_{ji} \left( x_{ji} - x_{j(i-1)}^2 \right)^2 
  \right\},
\end{equation}

where $\mu, x_{ji} \in \mathbb{R}$, $a, b_{ji} \in \mathbb{R}^+$, and where the final dimension of the distribution is given by the formula $n = (n_1 - 1)n_2 + 1$.
We choose $a = 30, b_{j,i} = 20, \mu=1$, and $n_2 = 3, n_1 = 4$, which yields a dimension of $n=10$. 
We use the \texttt{hybrid\_rosenbrock} library, which implements the HRD in \texttt{Python} \cite{leviyev2025hybrid}.

\subsection{Fisher preconditioning}
\label{subsec:fisher-precond-exp}
Our first experiment tracks the energy statistic \cite{panda2020hyppo,szekely2013energy} for the standard Langevin flow and the Fisher preconditioned Langevin flow.
We draw $N=200$ samples from a uniform distribution $\mathrm{Unif}[-5,5]^{\times 10}$, disable birth-death, and set the Fisher damping $\lambda=0.001$.
Note that the domain in this experiment is unconstrained.
For the preconditioned run, we set the timestep to $\tau=2$, and for the standard run we set the timestep to $\tau=0.001$.
Note that the timestep for the standard run was chosen such that the flow was stable.
On the other hand, the preconditioned flow is significantly more stable, and hence can be simulated with a larger timestep.
The results of the experiment are illustrated in the left panel of \cref{fig:precond}.
We see that the preconditioned flow outperforms the standard flow, converging around five times faster than its counterpart in $\epsilon$-statistic \cite{panda2020hyppo}, which is a two sample similarity test a set of i.i.d samples from from the target and the ensemble.

\subsection{Reparameterization Experiments}
\label{sec:experiments-reparam}
The aim of this section is to consider the effect of different reparameterizations. 
We keep all settings identical to \cref{subsec:fisher-precond-exp}, and utilize Fisher preconditioning.
However, we now constrain to the domain $[-5,5]^{\times 10}$.
Every setting is kept fixed, except for the reparameterization.
In the right panel of \cref{fig:precond} we illustrate the convergence utilizing the Gaussian, Logistic, and Cauchy reparameterizations.
Our results indicate that the Gaussian reparameterization converges about $2000$ iterations faster than the logistic, whilst the Cauchy lags behind.
We believe this speedup is attained since the gradient is mildly scaled down compared to the logistic case, whilst the restoring force in the Gaussian case better ensures particles do not get stuck in extremal regions close to the boundary.
These details are dicussed more in \cref{app:reparameterization}.
This provides a simple way of attaining faster converging dynamics, and may prove useful in situations where the forward model is expensive.
Finally, we note that the fact that the Gaussian reparameterization uses special functions to evaluate the pushforward and pullback maps is not a significant impediment to performance or implementation.

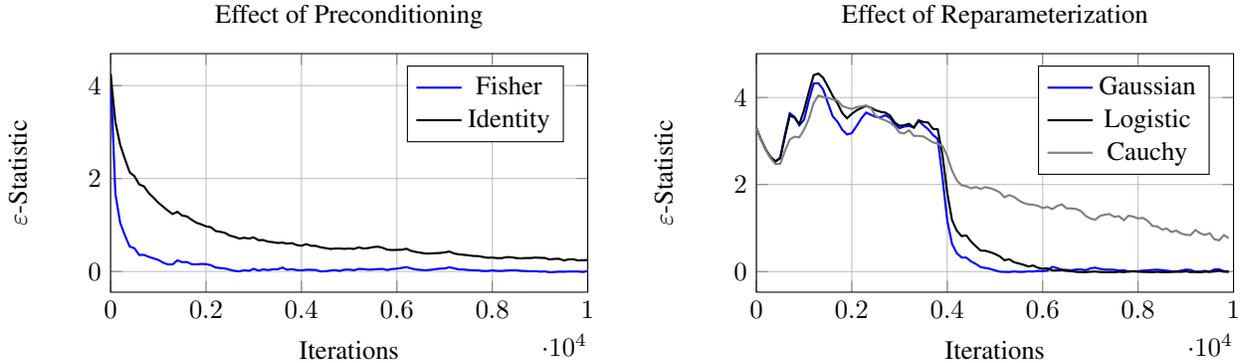
\begin{figure}[t]
  \centering

  \begin{minipage}[t]{0.48\textwidth}
      \centering
      \begin{tikzpicture}
      \begin{axis}[
          xlabel={Iterations},
          ylabel={$\varepsilon$-Statistic},
          title={Effect of Preconditioning},
          xmin=0, xmax=10000,
          legend style={
              at={(0.95,0.95)},
              anchor=north east,
              draw=black,
              fill=white,
          },
          grid=major,
          width=\textwidth,
          height=0.6\textwidth,
      ]
      \addplot[blue, thick] 
          table[x=x, y=fisher, col sep=space, header=true] {figures/ks_plot_small.dat};
      \addlegendentry{Fisher}
      
      \addplot[black, thick] 
          table[x=x, y=identity, col sep=space, header=true] {figures/ks_plot_small.dat};
      \addlegendentry{Identity}
      \end{axis}
      \end{tikzpicture}
  \end{minipage}%
  \hfill
  \begin{minipage}[t]{0.48\textwidth}
      \centering
      \begin{tikzpicture}
      \begin{axis}[
          xlabel={Iterations},
          ylabel={$\varepsilon$-Statistic},
          title={Effect of Reparameterization},
          xmin=0, xmax=10000,
          legend style={
              at={(0.95,0.95)},
              anchor=north east,
              draw=black,
              fill=white,
          },
          grid=major,
          width=\textwidth,
          height=0.6\textwidth,
      ]
      \addplot[blue, thick]
          table[x=x, y=Gauss, col sep=space, header=true] {figures/energy_stats.dat};
      \addlegendentry{Gaussian}

      \addplot[black, thick]
          table[x=x, y=Logistic, col sep=space, header=true] {figures/energy_stats.dat};
      \addlegendentry{Logistic}

      \addplot[gray, thick]
          table[x=x, y=Cauchy, col sep=space, header=true] {figures/energy_stats.dat};
      \addlegendentry{Cauchy}
      \end{axis}
      \end{tikzpicture}
  \end{minipage}

  \caption{Convergence of Langevin dynamics with \emph{(left panel)} optimal Fisher preconditioning and \emph{(right panel)} different random variables for reparameterization. We plot the evolution of the energy 2-sample statistic \cite{panda2020hyppo,szekely2013energy} between the evolving ensemble and i.i.d samples from the target. (left panel) We see that the flow with Fisher preconditioning converges quickly, while the flow without preconditioning requires many more iterations to sufficiently explore the energy landscape. (right panel) Convergence of the Langevin flow with different random variables used for reparameterization. We see that the Guassian reparameterization converges about $2000$ iterations faster than the logistic reparameterization, while the Cauchy reparameterization trails behind.}
  \label{fig:precond}
\end{figure}

\subsection{Two-ring Gaussian mixture model}

\begin{figure}[ht]
  \centering

  \begin{minipage}[t]{0.32\textwidth}
    \centering
    \includegraphics[width=\linewidth]{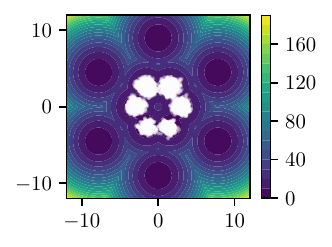}
    \caption*{(a) Standard}
  \end{minipage}
  \hfill
  \begin{minipage}[t]{0.32\textwidth}
    \centering
    \includegraphics[width=\linewidth]{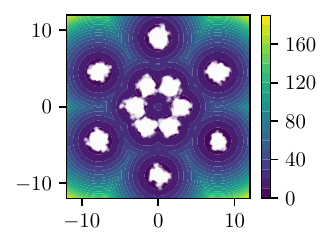}
    \caption*{(b) Annealed}
  \end{minipage}
  \hfill
  \begin{minipage}[t]{0.32\textwidth}
    \centering
    \includegraphics[width=\linewidth]{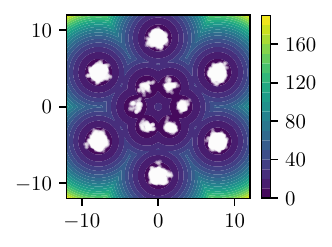}
    \caption*{(c) Annealed + BD}
  \end{minipage}

  \caption{Final samples produced by the Langevin dynamics in its standard implementation (\emph{left panel}), when augmented with annealing (\emph{central panel}), and augmented with both annealing and birth death (\emph{right panel}) on the two ring Gaussian mixture problem.}
  \label{fig:scatter_comparison}
\end{figure}

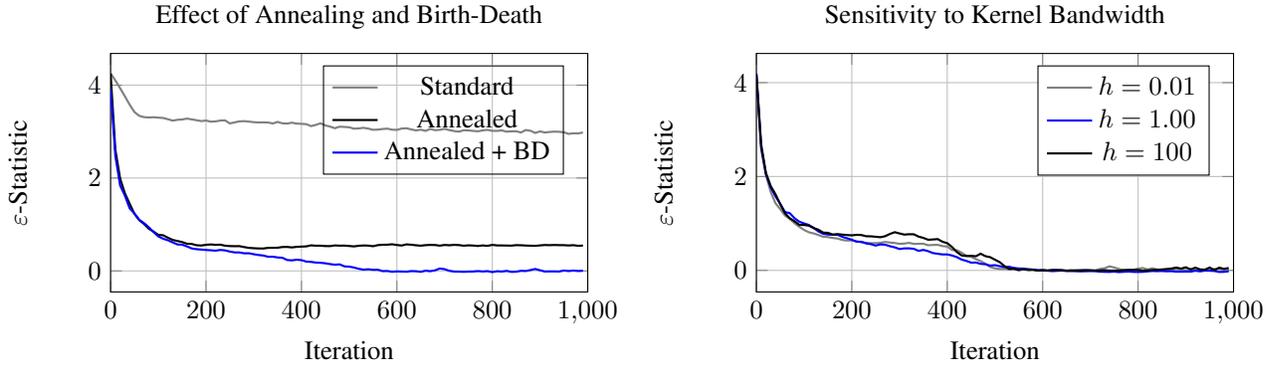
\begin{figure}[ht]
  \centering

  \begin{minipage}[t]{0.48\textwidth}
    \centering
    \begin{tikzpicture}
      \begin{axis}[
          width=\textwidth,
          height=0.6\textwidth,
          xlabel={Iteration},
          ylabel={$\varepsilon$-Statistic},
          grid=major,
          xmin=0, xmax=1000,
          legend style={at={(0.95,0.95)}, anchor=north east, draw=black, fill=none},
          title={Effect of Annealing and Birth-Death}
      ]
      \addplot[gray, thick] table [x=Iteration, y=Stat1, col sep=space] {figures/ring_convergence.dat};
      \addlegendentry{Standard}

      \addplot[black, thick] table [x=Iteration, y=Stat2, col sep=space] {figures/ring_convergence.dat};
      \addlegendentry{Annealed}

      \addplot[blue, thick] table [x=Iteration, y=Stat3, col sep=space] {figures/ring_convergence.dat};
      \addlegendentry{Annealed + BD}
      \end{axis}
    \end{tikzpicture}
  \end{minipage}
  \hfill
  \begin{minipage}[t]{0.48\textwidth}
    \centering
    \begin{tikzpicture}
      \begin{axis}[
          width=\textwidth,
          height=0.6\textwidth,
          xlabel={Iteration},
          ylabel={$\varepsilon$-Statistic},
          grid=major,
          xmin=0, xmax=1000,
          legend style={at={(0.95,0.95)}, anchor=north east, draw=black, fill=white},
          title={Sensitivity to Kernel Bandwidth}
      ]
      \addplot[gray, thick] table [x=iteration, y=hhh, col sep=space] {figures/ring_bandwidth_sensitivity.dat};
      \addlegendentry{$h=0.01$}

      \addplot[blue, thick] table [x=iteration, y=h, col sep=space] {figures/ring_bandwidth_sensitivity.dat};
      \addlegendentry{$h=1.00$}                                                  

      \addplot[black, thick] table [x=iteration, y=hh, col sep=space] {figures/ring_bandwidth_sensitivity.dat};
      \addlegendentry{$h=100$}

      \end{axis}
    \end{tikzpicture}
  \end{minipage}

  \caption{Convergence with respect to the $\epsilon$-statistic \cite{panda2020hyppo,szekely2013energy}. (\emph{left panel}) shows a comparison between the sampling methods. (\emph{right panel}) shows the sensitivity to bandwidth $h$.}
  \label{fig:ring_comparisons}
\end{figure}

We define a two-ring Gaussian mixture model in $\mathbb{R}^2$ composed of 12 Gaussian components: 6 on an inner ring and 6 on an outer ring.  
The inner ring has radius $r=3$ and total mixture weight of $0.1$, while the outer ring has radius $r=6$ and weight $0.9$.  
The geometry of the problem is illustrated in \cref{fig:scatter_comparison}.
We initialize $N=200$ samples from a standard Gaussian, and use a linear annealing schedule from $\beta_{\min} = 10^{-5}$ to $\beta_{\max} = 1$.
This test case is meant to simulate a phenomena encountered in GW problems.
Namely, we want to confirm that the particles can pass a shallow basin to reach the deeper basin.
In addition, we want to confirm that the mode weights can be recovered accurately.
Our results are illustrated in \cref{fig:ring_comparisons}.
We see that standard Langevin dynamics gets stuck at the local basin.
With annealing, the particles find all the modes, but are unable to reconstruct the weights of the modes.
Finally, with annealing and birth-death, all modes are found and balanced appropriately.
We also note that the results appear insensitive to the chosen bandwidth in \cref{eq:rbf-kernel}.
This is likely due to an imbalance between the energy and distance terms in the birth-death rate. 
A further investigation into optimal kernel choices for the birth-death process remains an interesting and open problem.

\section{Parameter estimation of GW150914}
\label{sec:gw150914}
Finally, we introduce and provide results for our primary application: the parameter estimation of the first detected binary black hole (BBH) coalescence, GW150914~\cite{Abbott2016}.
In the context of GW parameter estimation, our goal is to infer the source parameters $\theta \in \chi$ of a GW signal $s(t)$ given a signal embedded in noisy detector strain data $d(t) = h(t;\theta) + n(t)$.
Here, $n(t)$ denotes the noise process of the detector, and $h(t;\theta)$ is a signal model.
The \textit{Whittle likelihood} for a single detector is given by\footnote{See Ref.~\cite{Maggiore:2007ulw} for a pedagogical introduction.} 
\begin{equation}\label{eq:gw-potential}
    V(\theta) := -\log \mathcal{L}(\theta) = \frac{1}{2} \re \brk[a]1{d(\cdot) - h(\cdot \, ; \theta) | d(\cdot) - h(\cdot \, ; \theta)}_{\text{GW}},
\end{equation}
where the inner product $\brk[a]{\cdot|\cdot}_{\text{GW}}: \com^n \times \com^n \to \com$ is defined as
\begin{equation}\label{eq:gw-inner}
    \brk[a]{a,b}_{\text{GW}} \defn 4 \sum_{i=1}^n \frac{\tilde{a}(f_i)^* \tilde{b}(f_i)}{S_n(f_i)} \Delta f,
\end{equation}
where the tilde denotes a complex Fourier transform\footnote{In practice computed with a fast Fourier transform (FFT).}, $\Delta f = 1/T$ with $T$ the duration of the signal, and $S_n(f)$ is the one-sided noise power spectral density (PSD) of the detector which satisfies
\begin{equation}\label{eq:noisePSD}
    \brk[a]1{\tilde{n}(f) \tilde{n}(f')} = \frac{1}{2} S_n(f) \delta(f - f').
\end{equation}
Given a network of detectors, and under the assumption that the noise in each detector is uncorrelated, the network likelihood is simply given by the product of the individual detector likelihoods:
\begin{equation}\label{}
    \like_{\rm net}(\theta) = \prod_{i\in {\rm detectors}} \like_i(\theta).
\end{equation}

In general, the gravitational signal emitted by a (quasicircular) BBH system can be described in terms of a set of 15 parameters, $\theta = \{{\cal M}_c,\, q,\, \bm{\chi_1},\, \bm{\chi_2}, d_L,\, \theta,\, \phi,\, \iota,\, \psi,\, t_c,\, \Phi_c\}$ where:
${\cal M}_c = (m_1 m_2)^{3/5} / (m_1 + m_2)^{1/5}$ is the \emph{chirp mass} of the binary, $q = m_2/m_1$ is the \emph{mass ratio}, $m_1, m_2$ denote the component masses, $\bm{\chi_1}$ and $\bm{\chi_2}$ are the adimensional spins of the two black holes, $d_L$ is the luminosity distance to the source, $\theta$ and $\phi$ are the sky position parameters, $\iota$ is the inclination angle of the binary's orbital angular momentum with respect to the line of sight, $\psi$ is the polarization angle, $t_c$ is the time of coalescence, and $\Phi_c$ is the phase at coalescence.
In this paper, we employ an aligned-spin waveform model \textsc{IMRPhenomD} \cite{PhysRevD.93.044006,PhysRevD.93.044007} that outputs directly into frequency space and incorporates the inspiral, merger, and ringdown phases of the coalescence.
Note that this model only depends on the adimensional spin parameters aligned with the binary's orbital angular momentum, denoted by $\chi_{1z}$ and $\chi_{2z}$.
A summary of the parameters, their respective astrophysical priors\footnote{See \cite{callister2021thesauruscommonpriorsgravitationalwave} for a comprehensive review of GW priors.} and topological character is given in \cref{tab:param-topologies}.

\begin{table}[t]
  \centering
  \begin{tabular}{|c|c|c|c|c|}
  \hline
  \textbf{Parameter} & \textbf{Description} & \textbf{Support} & \textbf{Bounded/Periodic} & \textbf{Prior}\\
  \hline
  $\mathcal{M}_c$ & Chirp mass & $\mathbb{R}^+$ & B & Uniform in $(m_1, m_2)$ \\
  $q$ & Mass ratio & $(0, 1]$ & B & Uniform in $(m_1, m_2)$ \\
  $d_L$ & Luminosity distance & $\mathbb{R}^+$ & B & Power law\\
  $\theta$ & $\frac{\pi}{2}$ - Declination & $[0, \pi]$ & B & $\sin$\\
  $\phi$ & Right Ascension & $[0, 2\pi)$ & P & Uniform \\
  $\iota$ & Inclination angle & $[0, \pi]$ & B & $\sin$ \\
  $\psi$ & Polarization angle & $[0, \pi]$ & P & Uniform \\
  $t_c$ & Time of coalescence & $\mathbb{R}$ & B & Uniform \\
  $\phi_c$ & Phase of coalescence & $[0, 2\pi]$ & P & Uniform\\
  $\chi_1$ & Aligned spin 1 & $(-1, 1)$ & B & Isotropic \cite[Eq 5]{callister2021thesauruscommonpriorsgravitationalwave}\\
  $\chi_2$ & Aligned spin 2 & $(-1, 1)$ & B & Isotropic \\
  \hline
  \end{tabular}
  \caption{Parameters employed in the analysis of GW150914 with Descriptions, Support, and Type (Bounded/Periodic).}
  \label{tab:param-topologies}
\end{table}

We use 4~s of publicly available data \cite{abbott2021open} through the \texttt{gwpy} package~\cite{gwpy} around the GPS time 1126259462.4 for the two LIGO detectors, apply a Tukey window with roll-off parameter of 0.2~s as a preprocessing step, and use the \texttt{gwpy} implementation of the Welch method to estimate the PSD of the data in an interval of 32~s prior to the event (in this case with a roll-off of 0.4~s for the window).\footnote{These are the same settings employed in \url{https://git.ligo.org/lscsoft/bilby/blob/master/examples/gw_examples/data_examples/GW150914.py}.}
The resulting frequency grid is of size $F=1968$ over a range $[20, 512]$~Hz.
We use \texttt{gwfast}'s \cite{Iacovelli-2022} differentiable implementation of \textsc{IMRPhenomD} and the projection onto the detectors to calculate the network potential needed in the birth-death process, and use algorithmic differentiation to calculate the gradients for the Langevin flow.
The mirroring, preconditioning, birth-death, and topological adaptations are all incorporated as this task necessitates every one of these features.
For our algorithm, we use timestep $\tau=0.5$, stride $M=100$, iterations $L=20000$, $\gamma=0.01\eps$, $\sigma=0.01$. 
\cref{fig:cornerplot-gw} illustrates the results of this experiment.
In blue we have samples drawn from the \texttt{parallel Bilby} software \cite{Ashton-2019,Smith-2020} using the nested sampling algorithm \texttt{Dynesty} \cite{Speagle-2020}, and in red we have samples drawn from our method.
The \texttt{parallel Bilby} run is performed with the same data and waveform mode, employing 2048 live points and using the acceptance-walk sampling method for \texttt{Dynesty}.
The only difference between the two methods is represented by the use of a likelihood marginalized over $d_L, t_c, \Phi_c$ in the \texttt{parallel Bilby} run, employed to reduce the dimensionality of the problem and speed up the calculation. 
Considering this simplification, the \texttt{parallel Bilby} run took approximatley 1.5 hours to complete on a CPU cluster, parallelizing over 50 AMD EPYC 7742 CPUs.
After $2,775,1224$ likelihood evaluations, the effective number of samples was 9763. 

For Langevin birth-death, we set the timestep $\tau=0.5$, use the median bandwidth selection technique, and set the maximum fraction of teleporting particles to $f=0.05$. 
The number of particles is set to $N=500$, and the dynamics was run for $L=20000$ iterations with Fisher preconditioning, and a linear annealing schedule with $\beta_{\min} = 10^{-5}$.
We see in \cref{fig:cornerplot-gw} that the dynamics recovers the modes of the posterior well, but systematically overconstrains the parameters.
Note that in our method we do not incorporate Metropolis-Hastings corrections in the diffusion, and this naturally makes the diffusion biased \cite{vempala2019rapid}.
Another potential source of error is the kernel in the birth-death process, as kernel methods are highly sensitive to tuning parameters.
Our method took approximately 15 minutes to complete, where all model and gradient calculations were performed on an NVIDIA RTX-A6000 GPU.
Although the total number of model evaluations is approximately 10x larger than for parallel Bilby, we note that our evaluations are performed in parallel, and therefore the total time is still significantly reduced.
We believe with further refinement, representative samples may be routinely and robustly recovered with $N=200$ and $L=5000$ iterations.
Further investigation into improving the quality and robustness of recovered samples is left to future work.

\begin{figure}
    \centering
        \includegraphics[width=0.95\linewidth]{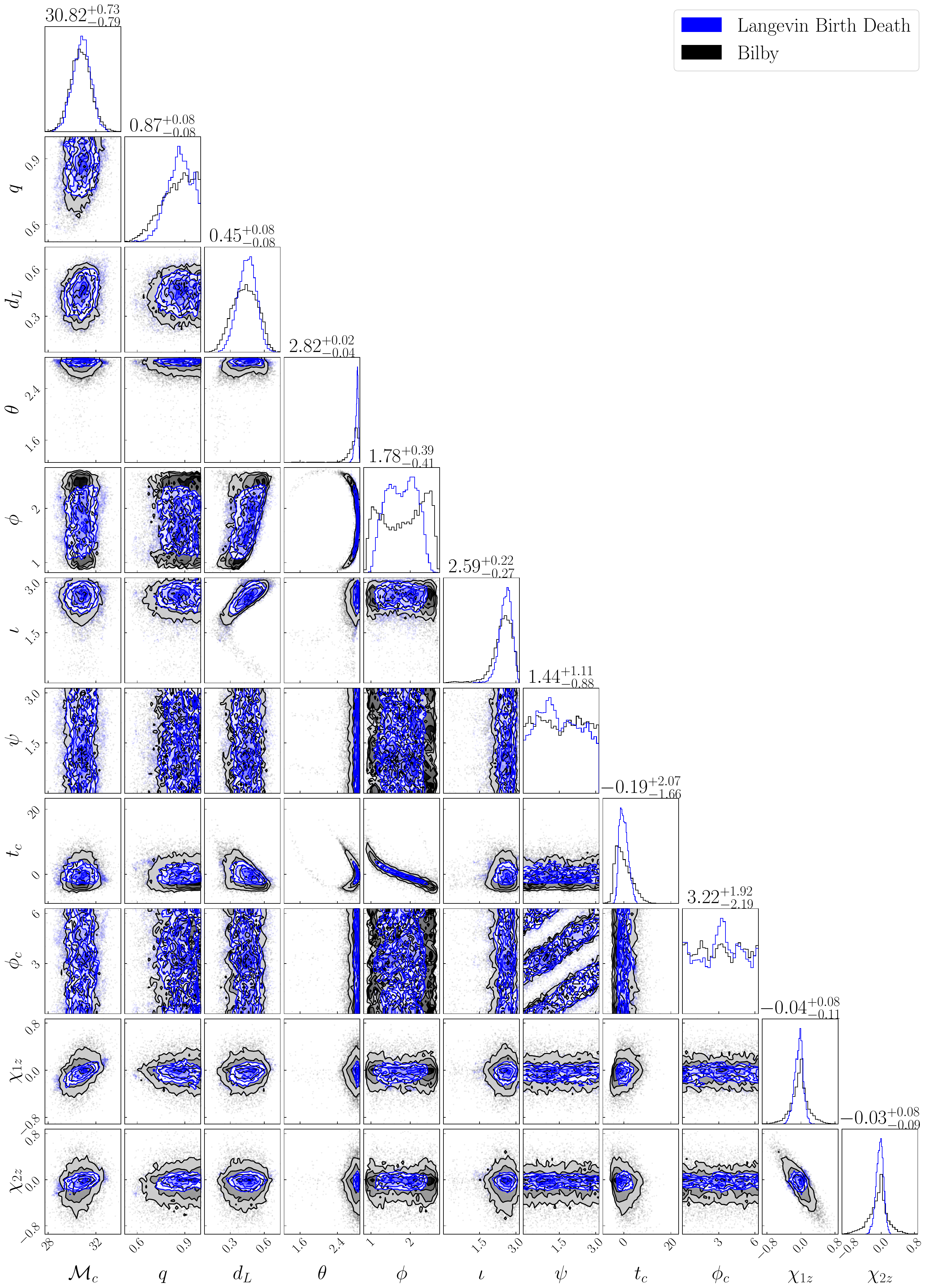}
    \caption{A corner-plot comparing samples obtained from the birth-death process (blue) and samples obtained from \texttt{Bilby} (black) for the GW event GW150914. We see that both results are in reasonable agreement.}
    \label{fig:cornerplot-gw}
\end{figure}

\section{Related Works and Conclusion}
The aim of our work has been to address issues related to adapting the Langevin dynamics to the particular needs of GW astronomy.
To this end, we proposed a multi-particle, pre-conditioned, and annealed Langevin diffusion augmented with a birth-death process that respects the topology of the hypercube and hypertorus.

In this paper, we simulate an ensemble of particles undergoing Langevin dynamics.
Notably, at any instance in time, all particles explore the posterior at temperature $1/\beta$.
Particles then teleport depending on rates calculated from \cref{eq:pampel-lambda}.
This is in contrast to replica exchange \cite{Syed2021}, where each particle explores the posterior at different temperatures, and swap probabilities are proportional to the energy and temperature discrepancies.
The advantage of replica exchange is that it utilizes a Metropolis filter, and is therefore asymptotically exact.
Furthermore, since there are particles that explore high temperature versions of the posterior throughout the simulation, this method is robust to mode discovery and retrainment.
However, since only one particle explores the posterior at room temperature, the dynamics must necessarily be simulated for longer in order to collect the desired quantity of samples.
The main tradeoff between birth/death and replica exchange is therefore speed vs robustness.
Both methods are embarrassingly parallel with respect to potential queries.
An asymptotically exact birth-death jump scheme has also been proposed recently \cite{lindsey2022ensemble}, but notably is a serial method.

In \cref{sec:reparam}, we adapted the Langevin dynamics to hypercube support spaces.
More generally, the Langevin dynamics may be applied to constrained spaces via mirrored Langevin dynamics (see \cref{app:mirror}), for which the hypercube is a special case.
Our aim, however, focused on analyzing the numerical behavior and properties of the chosen reparamaterization in this special case.
Consequently, we are able to provide guidelines on how to select the reparametrization, and suggest how to extend techniques like annealing to the hypercube.

In \cref{sec:ula}, we discussed the importance of preconditioning the Langevin dynamics, and in \cref{sec:fisher} provide a simple way to do so in the ensemble case.
Directions for further research include quasi-Newton methods, which are a generalization that allow for pointwise variation of the preconditioners.
These are expected to explore the parameter space more effectively, but come with their own set of challenges.
For example, the dynamics require calculating a costly correction term \cite{ma2015completerecipestochasticgradient}.
Alternatively, a Metropolis-Hastings filter may be incorporated.
See \cref{app:hessian-reparam} for a discussion on adapting quasi-Newton methods to the hypercube, and \cref{app:quasi} discussing quasi-Newton methods for GW parameter estimation.

\paragraph{Acknowledgements}
We are grateful to Michele Mancarella for his valuable insights, and to Konstantin Sarichev for generously providing GPU resources during the early stages of this project, and Jacob Lange for reviewing a draft of this work.
AZ was supported by NSF Grants PHY-2207594 and PHY-2308833, and by a College of Natural Sciences Catalyst Grant at UT Austin while carrying out this work.
F.I. is supported by a Miller Postdoctoral Fellowship and by NSF Grants No. AST-2307146, PHY-2513337, PHY-090003 and PHY-20043, by NASA Grant No. 21-ATP21-0010, by the John Templeton Foundation Grant 62840, by the Simons Foundation, and by the Italian Ministry of Foreign Affairs and International Cooperation grant No.~PGR01167.
This material is based upon work supported by NSF's LIGO Laboratory which is a major facility fully funded by the National Science Foundation.
This work has preprint numbers UT-WI-20-2025 and LIGO-P2500456. 

\printbibliography

\appendix

\section{Efficient Birth-Death Implementation}
\label{app:pt}
The birth-death process described in \cref{algo:pampel-bd} utilizes an abstract data structure (DS) with three methods: lookup, choice, and removal.
In this section we describe a specific \texttt{JIT} compilable implementation of this DS called \texttt{ParticleTracker} that allows us to implement the birth-death process with $\order(N)$ logical operations.
Given an ensemble $\{x_n\}_{n=1}^N$, we initialize two integer arrays of size $N$.
The first array, \texttt{particles}, contains labels for each particle in the ensemble.
The second array, \texttt{indices}, tracks the index at which a particle label appears in \texttt{particles}.  
As particles are ``killed'', the labels in \texttt{particles} are swapped to the right end of the array.
Similarly, the indices in \texttt{indices} are swapped to properly keep track of label locations in \texttt{particles}.
The trick is to group particles that are ``alive'' on the left side of the array: this makes lookup and choice trivial.
For lookup, if the index of a particle is less than or equal to $n_{\text{alive}} - 1$ inclusive, then it is alive, otherwise it is dead.
To randomly select a particle that is ``alive'', we select a random integer between $0$ and $n_{\text{alive}} - 1$ and choose the label correponding to that index in \texttt{particles}.
The \texttt{kill} method is illustrated in \cref{fig:bd-kill}, and a Python implementation for the \texttt{ParticleTracker} class can be found in \cite{repo}.

\begin{figure}[t]
  \centering
  \begin{tikzpicture}
      \node at (7, 4.2) {Particles};
      \node at (13, 4.2) {Indices};
      
      \newcommand{\drawArray}[5]{ 
          \node at (#1-1.5, #2) {#4}; 
          \foreach \i/\val/\col in {#3} {
              \draw[fill=\col, opacity=0.3] (#1+\i, #2-0.3) rectangle (#1+\i+1, #2+0.3); 
              \node at (#1+\i+0.5, #2) {\val}; 
          }
      }

      \drawArray{4.5}{3.5}{0/0/white, 1/1/white, 2/2/white, 3/3/white, 4/4/white}{Init};
      \drawArray{10.5}{3.5}{0/0/white, 1/1/white, 2/2/white, 3/3/white, 4/4/white}{};

      \drawArray{4.5}{2.5}{0/0/white, 1/1/white, 2/4/white, 3/3/white, 4/2/gray}{$\rightarrow$ \texttt{kill(2)}};
      \drawArray{10.5}{2.5}{0/0/white, 1/1/white, 2/4/white, 3/3/white, 4/2/white}{};

      \drawArray{4.5}{1.5}{0/0/white, 1/1/white, 2/3/white, 3/4/gray, 4/2/gray}{$\rightarrow$ \texttt{kill(4)}};
      \drawArray{10.5}{1.5}{0/0/white, 1/1/white, 2/4/white, 3/2/white, 4/3/white}{};

  \end{tikzpicture}
  \caption{Illustration of how the \texttt{kill} method works. An ensemble of size $N = 5$ is initialized, and the output of two consectutive kill operations are recorded. Dead particles are shaded in light grey.}
  \label{fig:bd-kill}
\end{figure}

\section{Reparameterization Analysis}
\label{app:reparameterization}

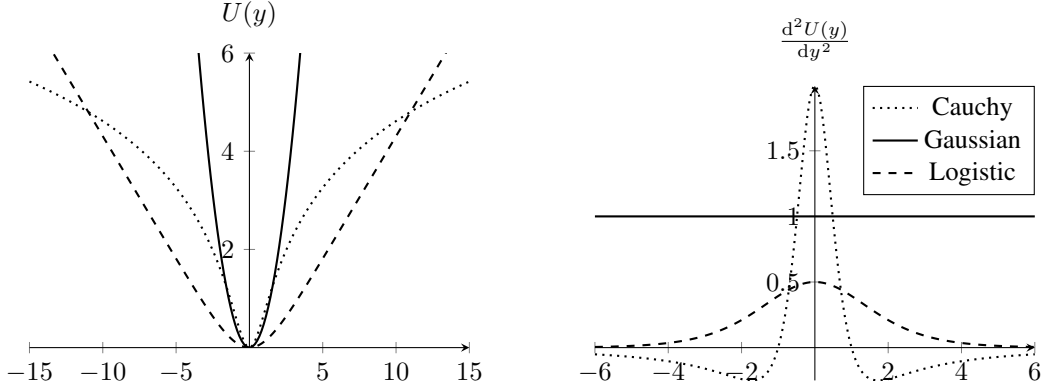
\begin{figure}[t]
    \centering

    \begin{subfigure}[t]{0.45\textwidth}
        \centering
        \begin{tikzpicture}
            \begin{axis}[
                domain=-15:15,
                samples=200,
                axis lines=middle,
                height=5.5cm,
                width=\linewidth,
                ymin=0, ymax=6,
                title={$U(y)$},
                legend style={at={(0.97,0.97)}, anchor=north east}
            ]
                \addplot[black, thick, dotted] { ln(1 + x^2) };
                \addplot[black, thick] { x^2/2 };
                \addplot[black, thick, dashed] { x/2 + ln(1 + exp(-x)) - ln(2) };
            \end{axis}
        \end{tikzpicture}
    \end{subfigure}
    \begin{subfigure}[t]{0.45\textwidth}
        \centering
        \begin{tikzpicture}
            \begin{axis}[
                domain=-6:6,
                samples=200,
                axis lines=middle,
                height=5.5cm,
                width=\linewidth,
                title={$\diff[2]{U(y)}{y}$},
                legend style={at={(1.00, 1.00)}, anchor=north east}
            ]
                \addplot[black, thick, dotted] {-2 * (x^2 - 1) / (1 + x^2)^2 };
                \addlegendentry{Cauchy}
                \addplot[black, thick] { 1 };
                \addlegendentry{Gaussian}
                \addplot[black, thick, dashed] { 1 / (1 + cosh(x)) };
                \addlegendentry{Logistic}
            \end{axis}
        \end{tikzpicture}
    \end{subfigure}

    \caption{An illustration of the confining potential and its second derivative for the Cauchy, Gaussian, and Logistic distributions.}
    \label{fig:confining-potentials}
\end{figure}

Our goal in this section is to analyze some of the consequences of \cref{thm:reparam}.

\subsection{Extremal Behavior}
From a high level, \cref{thm:reparam} indicates that the correction to the pushforward potential is a confining term which grows in strength close to the boundaries of the hypercube $\hyperc^d$.
Consequently, the gradient of the pushforward exhibits a repulsive force pointing in the opposite direction to the boundary.
This repulsive force assumes different forms depending on the particular choice of random variable (r.v) used to construct the reparameterization.
In this section, we will analize three cases using different r.v's to construct pushforward potentials.
Namely, we pick a sub-Gaussian r.v, a Gaussian r.v, and a r.v with heavier tails than a Gaussian.
We will also discuss what effect the particular choice has on gradient based dynamical systems close to boundary of the hypercube $\hyperc^d$.
\begin{example}[Logistic]
    \label{ex:logistic}
A standard logistic r.v.\ has the CDF  
\begin{equation}
F(y;0,1) = \frac{1}{1 + e^{-y}},
\end{equation}
and is the gold standard for reparameterization strategies in popular libraries such as Stan \cite{carpenter2017stan}, which uses the logit (or log-odds function) to transform coordinates supported on $[0,1]$.  
This correspondence becomes apparent upon observing that the inverse of the logistic CDF is the log-odds function, $F^{-1} = \logit$.
The PDF of the logistic r.v.\ is given by
\begin{equation}
f(y) = \frac{e^{-y}}{(1 + e^{-y})^2},
\end{equation}
where the normalizing constant in this case is unity.
The energy penalty hence takes the form  
\begin{align*}
- \sum_{i=1}^d \log f(y_i) 
&= - \sum_{i=1}^d \left[ -y_i - 2 \log(1 + e^{-y_i}) \right] \\
&= \sum_{i=1}^d \left[ y_i + 2 \log(1 + e^{-y_i}) \right].
\end{align*}
However, we are more interested in the asymptotic behavior of this penalty term.
Observe that
\begin{align*}
\log(1 + e^{-y_i}) 
&\approx 
\begin{cases}
e^{-y_i} & \text{if } y_i \gg 0 \\
- y_i & \text{if } y_i \ll 0.
\end{cases}
\end{align*}
Hence for $y_i \gg 0$, the penalty simplifies to  
\begin{align*}
\sum_{i=1}^d \left[ y_i + 2e^{-y_i} \right] 
&\approx \sum_{i=1}^d y_i.
\end{align*}
Similarly, for $y_i \ll 0$, we get  
\begin{align*}
\sum_{i=1}^d \left[ y_i + 2y_i \right] 
&= - \sum_{i=1}^d y_i.
\end{align*}
Asymptotically, the penalty in both cases can be summarized as $\sum |y_i|$.  
We find that using the logistic r.v.\ to reparameterize the measure is asymptotically equivalent to $\ell_1$ regularization (Lasso \cite{wang2013tikhonov}) of the energy.  
The repulsive correction then takes the form  
\begin{align*}
\partial_{y_i} \left[ - \log f(y_i) \right] 
&= 1 + 2 \cdot \frac{e^{-y_i}}{1 + e^{-y_i}} \cdot (-1) \\
&= 1 - 2 \cdot \frac{1}{1 + e^{y_i}} \\
&= 1 - 2 \cdot \frac{e^{-y_i}}{1 + e^{-y_i}} \\
&= 1 - 2 \left[ 1 - \frac{1}{1 + e^{-y_i}} \right] \\
&= 2F(y_i) - 1.
\end{align*}
Note that the repulsive correction asymptotes to $\pm 1$ at the left/right boundaries respectively. 
In addition, the logistic PDF is known to be sub-Gaussian, and therefore have strongly decaying tails.
This suggests that near the boundaries, since the gradient is dampened by the PDF, that the particles may have more difficulty moving near the boundaries.
\end{example}

\begin{remark}[Numerical stability]
Since we cannot rule out modes near the boundaries, we require our implementation of the reparameterization be robust to evaluations at extreme values.
Otherwise, we can expect $\infty$ and NaN related issues.  
We note that this framework is convenient since stable and tested implementions of the quantile, CDF, and PDF functions for most continous random variables are readily available in standard libraries such as \textsc{SciPy} \cite{virtanen2020scipy} or \textsc{JAX} \cite{bradbury2021jax}.
Therefore most, if not all, of the reparameterization machinery can be implemented without concerns about numerics.
\end{remark}

The logistic r.v.\ is an example of a sub-Gaussian r.v., meaning that the tails decay faster than a Gaussian.
As a natural next step, we consider what happens when the gradient is suppressed by a Gaussian PDF.
\begin{example}[Gaussian]
For a standard Gaussian r.v.\ the energy penalty term is given by  
\begin{align*}
-\sum_{i=1}^d \ln f(y_i) 
&= - \sum_{i=1}^d \ln \left[ \frac{1}{z} e^{- y_i^2 / 2} \right] \\
&= - \sum_{i=1}^d \left[ - \ln z - \frac{y_i^2}{2} \right] \\
&= \ln z^d + \frac{1}{2} \sum_{i=1}^d y_i^2.
\end{align*}
This result has a convenient physical interpretation.
Up to arbitrary constants, the penalty term is a harmonic confining potential. 
As a consequence, particles experience a ``Hooke's law'' type restoring force towards ``equilibrium'' 
\begin{equation}
    -\partial_{y_j} \ln f(y_j) = y_j.
\end{equation}
Alternatively, the penalty can be thought of as an $l_2$ (Tikhonov) regularization term.
The numerical advantage of this choice is that particles are more strongly repelled away from the boundaries.
\end{example}

We have seen an example of a sub-Gaussian and a Gaussian r.v.  
The final case study considers a heavier tailed r.v.
\begin{example}[Cauchy]
For a Cauchy r.v, the PDF is given by
\begin{equation}
    f(y_i) = \frac{1}{\pi} \cdot \frac{1}{1 + y_i^2}.
\end{equation}
Consequently, the energy penalty takes the form
\begin{align*}
    -\sum_{i=1}^d \ln f(y_i) &= -\sum_{i=1}^d \left[ -\ln \pi - \ln(1 + y_i^2) \right] \\
    &= \ln \pi^d + \sum_{i=1}^d \ln(1 + y_i^2).
\end{align*}
Similarly to \cref{ex:logistic}, we are interested in the asymptotic properties of this penalty.
A quick calculation yields the following
\begin{align*}
    \ln(1 + y_i^2) &= \ln \left[ y_i^2 \left(1 + \frac{1}{y_i^2} \right) \right] \\
    &= 2 \ln |y_i| + \ln \left(1 + \frac{1}{y_i^2} \right) \\
    &\sim 2 \ln |y_i|.
\end{align*}
Hence, the Cauchy distribution corresponds to a logarithmic penalty in the energy.  
The corresponding repulsive term is given by
\begin{equation}
    -\partial_{y_i} \ln f(y_i) = \frac{2 y_i}{1 + y_i^2}.
\end{equation}
Such a reparameterization may be useful in situations where modes are expected to be close to the boundaries, since the repulsive force tends to zero asymptotically, and the motion is entirely controlled by the scaled gradient.
\end{example}

\subsection{Hessian Reparameterization}
\label{app:hessian-reparam}
In this paper we utilized a homogeneous (constant) preconditioning matrix that was directly constructed on the dual space.  
It is natural to consider, however, whether one may inherit preconditioners on the dual space space from the hypercube, as this would provide a path to adapting quasi-Newton methods to bounded spaces.
In this section, we provide preliminary results in this direction.

Given \cref{thm:reparam}, we see that the pushforward Hessian takes the form
\begin{equation}
\nabla^2 V_\#(y) = \nabla^2 V(T^{-1}(y)) \odot (\Delta f(y) \otimes \Delta f(y)) + \nabla^2 U(y) + \text{diag} \left( \nabla V(T^{-1}(y)) \odot \Delta \odot f'(y) \right),
\end{equation}
where $\Delta := (\Delta_1, \ldots, \Delta_d)$, $f(y) := (f_1(y_j), \ldots, f_d(y_j))$, and $\odot, \otimes$ denote the Hadamard and outer product operations respectively.  
It is crucial for the preconditioner to be positive definite (p.d) in order to guarantee a descent direction.  
Suppose $\nabla^2 V$ is convex or replaced with some p.d. approximation thereof.  
Then the first term is guaranteed to be p.d. by Schur's product theorem.  
The definiteness of the second term depends on the r.v. used to construct the transformation, as illustrated in \cref{fig:confining-potentials}.  
The final term is a bulk correction whose sign depends on whether the driving and repulsive forces in the pushforward gradient are aligned or anti-aligned.  
This is a consequence of the fact that
\begin{equation}
\partial_j f(x_j) = f(x_j) \, \partial_j \ln f(x_j) = - f(x_j) \, \partial_j U(x_j)
\end{equation}
Because this term depends on $V$ and is negative when the two forces are aligned, it is not trivial to understand what effect it has on the definiteness of the Hessian in general, and what effect modifying this term would have on the dynamics.  
We leave a more detailed analysis of this question for future work.  
However, the following example provides intuition on the pushforward Hessian in a toy case.
\begin{example}[Harmonic]
    Consider the bounded harmonic potential
    \begin{equation}
    V(x) = \frac{1}{2} x^2, \quad x \in [0,1],
    \end{equation}
    which has the convenient property that its Hessian is positive-definite everywhere and equal to unity.  
    Then the pushforward Hessian takes the form
    \begin{align}
    \frac{d^2 V_\#(y)}{dy^2} &= \frac{d^2 V(T^{-1}(y))}{dy^2} \, f^2(y) + \frac{d^2 U(y)}{dy^2} + V'(T^{-1}(y)) \, f'(y) \nonumber \\
    &= f^2(y) + \frac{d^2 U(y)}{dy^2} + T^{-1}(y) f'(y). \label{eq:harmonic-push-pot}
    \end{align}
    We evaluated this expression for the Cauchy, Gaussian, and logistic cases using a CAS, and plotted them in \cref{fig:hessian-reparam}.  
    There are several notable conclusions we can draw.  
    First, the pushforward inherited by the Cauchy case is negative definite for the majority of its support, and in its present form may not be used as a preconditioner.  
    This is in spite of the fact that the Hessian was positive-definite everywhere in the hypercube.  
    Second, we see that the pushforward Hessian approaches zero near the boundaries of the hypercube.  
    A Newton scheme would thus yield larger effective stepsizes near the boundary, but would also run the risk of becoming singular.  
    A damping term appears to be required to avoid this possible numerical instability.  
    Interestingly, the Gaussian case does not appear to suffer from this.  
    Finally, although the alignment term is not strictly positive everywhere, the pushforward Hessian remains positive everywhere in the Gaussian and logistic cases.

\end{example}

\begin{figure}
    \centering
    \begin{tikzpicture}
    \begin{axis}[
        width=8cm,
        height=6cm,
        xlabel={$y$},
        title={Comparison of $\frac{d^2 V_\#(y)}{dy^2}$},
        legend style={at={(0.97,0.97)}, anchor=north east},
        grid=both
    ]
    
    \addplot [black, thick] table [x index=0, y index=1, col sep=space] {figures/hessian_reparam.dat};
    \addlegendentry{Gaussian}
    
    \addplot [black, thick, dashed] table [x index=0, y index=2, col sep=space] {figures/hessian_reparam.dat};
    \addlegendentry{Logistic}
    
    \addplot [black, thick, dotted] table [x index=0, y index=3, col sep=space] {figures/hessian_reparam.dat};
    \addlegendentry{Cauchy}
    
    \end{axis}
    \end{tikzpicture}
    \caption{Here we plot \cref{eq:harmonic-push-pot} for the Gaussian, Logistic, and Cauchy cases.}
    \label{fig:hessian-reparam}
\end{figure}
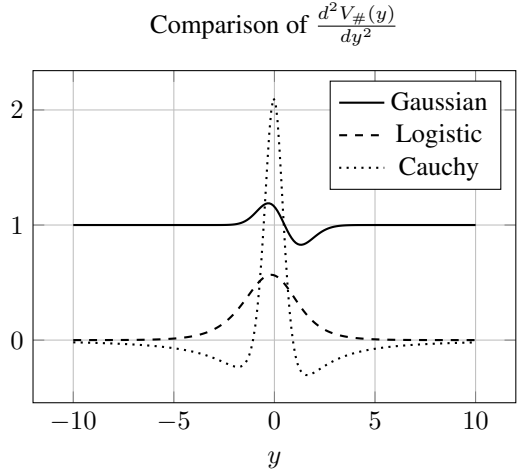

\subsection{Relationship to Mirrored Langevin Diffusion}
\label{app:mirror}

Given support space $\chi$, a \textit{mirror map}\footnote{See \cite{shi2021sampling} for additional technical conditions $\psi$ must satisfy.} $\psi$ is chosen such that $\nabla \psi$, the \textit{pushforward}, maps to unbounded Euclidean space (i.e $\nabla \psi(\chi) = \reals^d$).
In the mirrored descent literature \cite{hsieh2020mirroredlangevindynamics,li2021mirrorlangevinalgorithmconverges,shi2021sampling}, $\chi$ is referred to as the \textit{primal} space, and $\reals^d$ is referred to as the \textit{dual space}.
The inverse map, also called the \textit{pullback}, is obtained through $(\nabla \psi)^{-1} = \nabla \psi^*$, where $\psi^*$ is the \textit{Legendre-Fenchel dual}\footnote{Note that this is a generalization of the standard Legendre transform in physics to a larger class of functions (i.e, non-differentiable, etc.), and appears often in convex analysis.} of $\psi$, defined by 
\begin{align}
\psi^*(y) := \sup_{x \in \chi} y^{\top}x - \psi(x).
\end{align}

Given a potential $V:\chi \to \reals$ and an appropriate mirror map $\psi$, the \textit{Mirror Langevin dynamics} \cite{hsieh2020mirroredlangevindynamics} obey the following stochastic differential equation:
\begin{align}
dY_t &= - \nabla V_\#(Y_t) dt + \sqrt{2} dB_t, & X_t &= \nabla \psi^*(Y_t),
\end{align}
where $V_\#$ is the potential of density $p$ under pushforward $\nabla \psi$, given by 
\begin{align}
V_\#(y) = V(x) + \ln \det \nabla^2 \psi(x).
\end{align}

Upon inspection, this formulation of the mirrored Langevin dynamics is identical to running Langevin dynamics on an unbounded reparameterization of the original problem, and samples are collected by applying the pullback $\nabla \psi^*$ \cite[Theorem 1]{hsieh2020mirroredlangevindynamics}.
In this paper, we are particularly interested in the case where $\chi = \hyperc^d$, a hypercube.
We show in the next example that the logistic transformation arises from the sum of negative Shannon entropy.

\begin{example}[Logistic]
    Again, consider the hypercube $\hyperc^d = \bigtimes_{i=1}^d (a_i, b_i)$, and let us define the negative sum of Shannon entropy mirror map
    \begin{equation}
        \label{eq:shannon-mirror}
    \psi(x) = \sum_{i=1}^d \left[ (x_i - a_i)\ln(x_i - a_i) + (b_i - x_i)\ln(b_i - x_i) \right].
    \end{equation}
    
    Then the components of the gradient are given by
    \begin{align*}
    \partial_{x_i} \psi(x) 
    &= \ln(x_i - a_i) + \frac{(x_i - a_i)}{(x_i - a_i)} - \ln(b_i - x_i) + \frac{(b_i - x_i)}{(b_i - x_i)}(-1) \\
    &= \ln\left(\frac{x_i - a_i}{b_i - x_i}\right) \\
    &= \logit\left(\frac{x - a_i}{b_i - a_i}\right) \tag{*} \\
    &= \logit \circ \sigma_i^{-1}(x),
    \end{align*}
    where $\sigma_i(x) = (b_i - a_i)x + a_i$ is the affine transformation introduced in \cref{sec:reparam}, and we have used the fact that    
    \begin{align*}
    \frac{p}{1 - p} = \frac{x - a}{b - x}
    &\Rightarrow \quad p(b - x) = (1 - p)(x - a) \\
    &\Rightarrow \quad p[(b - x) + (x - a)] = x - a \\
    &\Rightarrow \quad p = \frac{x - a}{b - a}.
    \end{align*}
    Hence, the mirror map in \cref{eq:shannon-mirror} leads to the logistic reparameterization.
    The Legendre-Fenchel dual of the mirror map turns out to be
    \begin{align}
    \psi^*(y) = y a - (b-a) \ln \frac{b-a}{a+e^y},
    \end{align}
    and a simple calculation yields
    \begin{align} 
    \nabla \psi^*(y) = \sigma(F(x)),
    \end{align}
    as expected.
    The methods are in complete agreement.
    \end{example}

We highlight that the framework developed in \cref{sec:reparam} augments the theory of mirrored Langevin dynamics, as it makes the behavior of the flow near the boundary clearer, and assists in generalizing techniques such as annealing to bounded spaces.

\subsection{Bounding Birth-death Jumps}
\label{app:bounding-jumps}

To ensure fidelity of the decoupled jump procedure, we should track the number of particles that are jumping, and ensure that it is not a significant fraction of the ensemble.  
The following result illustrates that the expected number of particles that jump in a particular iteration is proportional to the mean absolute deviation of the rates.

\begin{proposition}
Define the random variable
\begin{equation}
    \xi = \sum_{i=1}^N \mathbbm{1}\left\{ r_i < 1 - e^{-|\Lambda_i|c} \right\},
\end{equation}
where each $r_i \sim \text{Unif}[0,1]$ and $c > 0$.  
If $\langle |\Lambda| \rangle c < 2$, then the expected value is bounded by $\mathbb{E}\, \xi \leq N f$, where the fraction $f := \frac{\langle |\Lambda| \rangle c}{2}$.
\end{proposition}

\begin{proof}
Let $x_i := |\Lambda_i|c$.  
Then
\begin{align*}
    \mathbb{E} \, \xi 
    &= \sum_{i=1}^N \mathbb{E} \, \mathbbm{1}\left\{ r_i < 1 - e^{-x_i} \right\} \\
    &= \sum_{i=1}^N \left(1 - e^{-x_i} \right) \\
    &\leq \sum_{i=1}^N \frac{x_i}{1 + x_i/2} & \text{(Holds for every $x_i > 0$)} \\
    &= N \left\langle \frac{x_i}{1 + x_i/2} \right\rangle \\
    &\leq \frac{N \langle x_i \rangle}{1 + \langle x_i \rangle / 2} & \text{(Reverse Jensen's inequality)} \\
    &\leq \frac{N \langle x_i \rangle}{2} & \text{(restrict to $\langle x_i \rangle < 2$)} \\
    &= N f. \qedhere
\end{align*}
\end{proof}

In practice, the constant $c$ may always be chosen such that this bound holds.  
This provides a useful diagnostic for the birth-death process.  
If $c \ll 1$ must be chosen to suppress the number of jumps, it is likely the rates need to be calculated with a different kernel.

\section{Quasi-Newton}
\label{app:quasi}
Quasi-Newton methods are an exciting direction to investigate further, as they have the potential to dramatically improve convergence of the Langevin dynamics. 
For a nonlinear least squares problem, a well-known approach is the Levenberg-Marquardt algorithm \cite{nocedal1999numerical}, which interpolates between a quasi-Newton and gradient descent type update. 
This method critically depends on the fact that the cost function is of least squares form, as this allows one to decompose the Hessian into a positive definite and a definite matrix. 
The positive definite component is referred to as the Gauss-Newton approximation, and is used in the Levenberg-Marquardt algorithm.

Interestingly, the Hessian of the Whittle likelihood given in \cref{eq:gw-potential} exhibits an analogous decomposition. 
We have:
\begin{align*}
    V(\theta)_{,ij} 
    &= \frac{1}{2} \re \left\langle h(\cdot \,; \theta) - d(\cdot) \mid h(\cdot \,; \theta) - d(\cdot) \right\rangle_{,ij} \\
    &= \re \left\langle h_{,i}(\cdot \,; \theta) \mid h(\cdot \,; \theta) - d(\cdot) \right\rangle_{,j} \\
    &= \re \left\langle h_{,i}(\cdot \,; \theta) \mid h_{,j}(\cdot \,; \theta) \right\rangle + \re \left\langle h_{,ij}(\cdot \,; \theta) \mid h(\cdot \,; \theta) - d(\cdot) \right\rangle
\end{align*}
where we define the first term as the Gauss–Newton approximation of $\nabla^2 V$, which we label $G$. 
The Gauss-Newton approximation $G$ can be shown to be positive definite with a simple calculation:
\begin{align*}
    x_i G_{ij} x_j 
    &= x_i x_j \, \re \left\langle h_{,i}(\cdot \,; \theta) \mid h_{,j}(\cdot \,; \theta) \right\rangle\\
    &= x_i x_j \, \re  \sum_n \frac{h_{,i}(f_n; \theta)^* h_{,j}(f_n; \theta)}{S(f_n)} \Delta f  \\
    &= x_i x_j \, \re \sum_n \left[ \frac{(h_{1,i}(f_n; \theta) - i h_{2,i}(f_n; \theta))(h_{1,j}(f_n; \theta) + i h_{2,j}(f_n; \theta))}{S(f_n)} \right] \Delta f\\
    &= x_i x_j \sum_n \left[ \frac{h_{1,i}(f_n; \theta) h_{1,j}(f_n; \theta) + h_{2,i}(f_n; \theta) h_{2,j}(f_n; \theta)}{S(f_n)} \right] \Delta f> 0
\end{align*}
since $S > 0$ and since the sum of positive definite matrices is positive definite as well.
There is an interesting connection between $G$ and what is known as the Fisher matrix $\Gamma$ in the GWPE literature \cite{vallisneri2008use}. 
Namely, the Fisher matrix is the Gauss-Newton matrix $G$ evaluated at the maximum a posteriori (MAP) point.
That is, $G(\theta_{\text{MAP}}) = \Gamma$.
We leave the investigation of quasi-Newton methods for future work.



\end{document}